\documentclass[a4paper]{article}

\usepackage{times}
\usepackage{amsmath}
\usepackage{amssymb}
\usepackage{amsthm}
\usepackage{stmaryrd}
\usepackage[active]{srcltx}

\newtheorem{thm}{Theorem}
\newtheorem{lm}[thm]{Lemma}
\newtheorem{df}[thm]{Definition}
\newtheorem{cor}[thm]{Corollary}

\numberwithin{equation}{section}
\numberwithin{thm}{section}

\theoremstyle{remark} \newtheorem*{rem}{Remark}

\newcommand{\Abar}{\overline{\mathcal {A}}}
\newcommand{\Af}{\mathfrak{A}}

\newcommand{\C}{\mathbb{C}}
\newcommand{\R}{\mathbb{R}}

\newcommand{\h}{{\cal H}}

\newcounter{mnotecount}[section]

\begin{document}
\title{New diffeomorphism invariant states on a holonomy-flux algebra\footnote{This is an author-created, un-copyedited version of an article accepted for publication in Classical and Quantum Gravity. IOP Publishing Ltd is not responsible for any errors or omissions in this version of the manuscript or any version derived from it. The definitive publisher authenticated version is available online at http://dx.doi.org/10.1088/0264-9381/27/22/225005}}
\author{ Micha{\l} Dziendzikowski$^a$, Andrzej Oko{\l}\'ow$^b$}
\date{December 7, 2009}

\maketitle
\begin{center}
{\it  Institute of Theoretical Physics, Warsaw University\\ ul. Ho\.{z}a 69, 00-681 Warsaw, Poland\smallskip\\ $^a\!$ mdzi@fuw.edu.pl, $^b\!$ oko@fuw.edu.pl}
\end{center}
\medskip

\begin{abstract}
The theorem by Lewandowski {\em et al.} stating uniqueness of a diffeomorphism invariant state on an algebra of quantum observables for background independent theories of connections is based on some technical assumptions imposed on the algebra and the diffeomorphisms. In this paper we present a class of diffeomorphism invariant states on an algebra of this sort, which exist when the algebra and the diffeomorphisms satisfy alternative assumptions. 
\end{abstract}

\section{Introduction}

As commonly known canonical quantization of a classical theory begins with a choice of elementary classical observables, that is, functions on the phase space of the theory which are supposed to have unambiguous quantum analogs. In the next step one uses the classical observables to generate a $*$-algebra of elementary quantum observables. Once a $*$-representation of the $*$-algebra on a Hilbert space  is chosen the kinematics of the corresponding quantum theory is thereby defined. Definition of dynamics requires construction of a Hamilton operator on the Hilbert space, but often it has to be preceded by defining constraint operators and finding all physical states (that is, states annihilated by the constraint operators). In general, the choice of elementary classical observables is not unique and there are also many $*$-representations of the algebra of quantum observables. Moreover, the definition of the Hamiltonian and constraint operators is also not always unique since it may e.g. involve regularizations. Thus, given a classical theory, a quantum model resulting from canonical quantization procedure is far from being unique.   

In this paper we address some aspects of the issue of uniqueness of quantum models obtained by canonical quantization of diffeomorphism invariant theories of connections\footnote{A diffeomorphism invariant theory of connections means  a theory in a Hamiltonian form such that $(i)$ its configuration space is a space of connections on a principal bundle $P(\Sigma,G)$, where $\Sigma$ is a base manifold, and $G$ is a Lie group (i.e. the structure group of the bundle and the connections), $(ii)$ there exist Gauss and vector constraints imposed on the phase space which ensures, respectively, gauge and diffeomorphism invariance of the theory \cite{ALMMT}.}. These aspects are  $(i)$ the choice of elementary classical observables and $(ii)$ the choice of a $*$-representation of the corresponding algebra of quantum observables. Our interest in this issue is motivated by the fact that the class of diffeomorphism invariant theories of connections includes the Einstein's theory of general relativity as its most important member. Thus the results presented in this paper may turn out to be relevant for the uniqueness issue of Loop Quantum Gravity (LQG) being an advanced attempt to apply the rules of canonical quantization to general relativity  \cite{ALMMT,BIQ,Thmnn}. 

The Ashtekar-Corichi-Zapata (ACZ) algebra \cite{ACZ} is commonly accepted as an algebra of elementary classical observables for diffeomorphism invariant theories of connections. The configuration part of the phase space of such a theory is the space $\cal A$ of connections on a principle bundle $P(\Sigma,G)$ (where $\Sigma$ is the base manifold, and $G$---the structure group), and the momentum part consists of vector densities on $\Sigma$ valued in the dual $\mathfrak{g}^*$ of the Lie algebra $\mathfrak{g}$ of the Lie group $G$. The ACZ algebra consists of so called cylindrical functions defined on the configuration space  and derivative operators called flux operators acting on the functions---the operators correspond to momentum degrees of freedom. Each cylindrical function is associated with some paths embedded in $\Sigma$, while each flux operator is specified by a choice of an oriented face $S\subset \Sigma$ with codimension $1$ equipped with a {\em smearing function} $f:S\to\mathfrak{g}$. It has to be emphasized that, given a theory, the ACZ algebra is still not unique, since there is some freedom in specifying the class of paths, faces and smearing functions. For example, given real-analytic $\Sigma$, the paths and the faces can be chosen to be either analytic or semianalytic{\footnote{{\em Semianalytic} means roughly {\em piecewise analytic}. For a precise definition see e.g. \cite{LOST}.}. Moreover, one can choose the smearing functions on the faces to be either smooth or continuous or one can require them to have compact supports or give up this requirement. 

The $*$-algebra of elementary quantum observables generated by the ACZ algebra is called {\em holonomy-flux algebra}. Here we are concerned with the issue of  uniqueness of its $*$-representations on Hilbert spaces. Of course, the question has to be narrowed to the class of irreducible or cyclic representations, but even then there seem to exist many inequivalent representations. Therefore one has to  apply some other criteria in order to restrict the number of admissible representations. An inescapable criterion is that of (suitably defined) covariance of a representation with respect to gauge transformations of the theory which are generated by the Gauss and vector constraints\footnote{From the viewpoint of bundle theory the gauge transformations are generated by automorphisms of the principal bundle. If the bundle is trivial then the group of bundle automorphisms is generated by the group of Yang-Mills gauge transformations (generated by the Gauss constraint) and the group of diffeomorphisms of the base manifold (generated by the vector constraint).}. If it is combined with the requirement that a representation is generated via the Gelfand-Naimark-Segal (GNS) construction from a state on the holonomy-flux algebra then covariance of the representation follows from invariance of the state with respect to the transformations. Then one can prove the following uniqueness theorem \cite{LOST}:

\begin{thm}
Assume that the base manifold $\Sigma$ is semianalytic. Let the paths and the faces defining  elements of the ACZ algebra be semianalytic, and the smearing functions on the faces are of compact support. Then on the resulting holonomy-flux algebra there exists a unique state invariant with respect to the gauge transformations which include semianalytic diffeomorphisms of $\Sigma$.   
\label{lost}
\end{thm}      

A claim that this theorem completely solves the uniqueness problem of the kinematical framework of quantum models under consideration would be rather unjustified. There are two main objections against such a claim:
\begin{enumerate}
\item as pointed out by Varadarajan \cite{Var} the theorem concerns invariant states and consequently uniqueness is established only among covariant {\em GNS representations} of the holonomy-flux algebra.
\item the theorem is based on the assumption of compact support of smearing functions, which seems to be rather technical. 
\end{enumerate}
It was shown in \cite{LOST} that in the case of a principle bundle with a {\em two-dimensional} oriented base manifold and the structure group $U(1)$ Theorem \ref{lost} loses its validity once we extend the ACZ algebra by admitting smearing functions with non-compact support: the state given by Theorem \ref{lost} admits then non-unique extensions. The authors of \cite{LOST} rise also a question whether these non-unique extensions are a peculiarity of the dimension of the base manifold equal $2$.

The goal of our research \cite{mich} was to check whether  in the case of a base manifold of dimension greater than $2$ the uniqueness theorem holds once we give up the requirement of compact support of smearing functions. In the present paper we show that if the ACZ algebra is built over the bundle $P=\R^D\times U(1)$, $D\geq 3$ and if
\begin{enumerate}
\item we allow the smearing functions to be supported also on some non-compact sets,
\item paths and faces defining elementary observables are {\em analytic} 
\end{enumerate}
then on the resulting holonomy-flux algebra there exists many distinct states invariant with respect to the gauge transformations which do not include any other diffeomorphisms of $\R^D$ than {\em analytic} ones. Our construction is perhaps valid also in the category of {\em semianalytic} paths and faces (and, consequently, semianalytic diffeomorphisms) but a proof is still lacking in this case.

The paper is organized as follows. In Section 2 we remind the construction of the ACZ and holonomy-flux algebras. In Section 3 we introduce an ansatz for a functional on holonomy-flux algebras, investigate its properties and finally prove that it actually defines diffeomorphism invariant states on an appropriately constructed holonomy-flux algebra. Section 4 contains a short discussion of the result.

\section{Elementary classical and quantum observables \label{sec-el}}

In the present section we shall briefly recall the construction of algebras of elementary and quantum observables, that is, the ACZ algebra and the holonomy-flux $*$-algebra, respectively \cite{ACZ,Thmnn,LOST}. For the sake of simplicity we assume the $G$-bundle $P$ to be trivial 
\[
P=\Sigma\times G,
\]
with a semianalytic $\Sigma$ and a {\em compact connected} Lie group $G$.      

\subsection{Cylindrical functions}

The properties of connections on a given principal bundle $P$ provide us 
with some natural algebra of functions separating points of the classical configuration space $\mathcal{A}$. To define these functions we use edges embedded in $\Sigma$:   

\begin{df}
(i) A semianalytic edge $e$ is a one-dimensional semianalytic embedded $C^{1}$ submanifold\footnote{For the definition of a semianalytic embedded submanifold see \cite{LOST}.} of $\Sigma$ with 2-point boundary.  (ii) An edge is a one-dimensional oriented embedded $C^{0}$ submanifold of $\Sigma$ with 2-point boundary given by a finite union of semianalytic edges. 
\end{df}

We say that two edges are independent if and only if they either $(i)$ have no common point or $(ii)$ every common point of them is an endpoint of both edges. 
\begin{df}
A graph $\gamma$ is a finite set $\gamma=\{e_{1},\ldots,e_{n}\}$ of pairwise independent edges.
\end{df}

In the case of the trivial bundle $P$ each parallel transport (holonomy) $A_{e}$ along an edge $e$ given by a connection $A\in\mathcal{A}$ can be unambiguously described by an element $h_e(A)$ of the structure group $G$. Let us denote as $\mathcal{A}_{e}\subset G$ the subset of $G$ obtained as holonomies of all smooth connections $A\in\mathcal{A}$ along the edge $e$. Since $G$ is connected $\mathcal{A}_{e}\simeq G$ \cite{al-hoop}. 

\begin{df}
  Given a graph $\gamma=\{e_{1},\ldots, e_{n} \}$ and a smooth function $\psi:G^{n}\rightarrow \mathbb{C}$ we define a smooth cylindrical function $\Psi:\mathcal{A}\rightarrow \mathbb{C}$ as
\begin{equation}
\Psi(A):=\psi(h_{e_{1}}(A),\ldots,h_{e_{n}}(A))\equiv h_{\gamma}^{*}\psi(A).
\end{equation}
We say that $\Psi$ is compatible with the graph $\gamma$. We denote by $Cyl^{\infty}$ the set of all smooth cylindrical functions.
\end{df}
One can easily see from the definition that any cylindrical function $\Psi$ is compatible with infinitely many graphs $\gamma$. In particular, given $\Psi_{1},\Psi_{2}$, there exists a graph $\gamma$ compatible with both, i.e $\Psi_{1}=h_{\gamma}^{*}\psi_{1}$ and $\Psi_{2}=h_{\gamma}^{*}\psi_{2}$. 

It turns out that $Cyl^\infty$ with the natural linear structure and the natural multiplication is an algebra. Elements of this algebra, that is, cylindrical functions are chosen to be classical configuration observables. Later we will see that, roughly speaking, the algebra $Cyl^\infty$ forms the 'holonomy part' of the holonomy-flux algebra.

One can also easily prove \cite{ai} that $Cyl^\infty$ can be equipped with a structure of a normed $*$-algebra with an involution 
\begin{equation}
\Psi^*(A):=\overline{\Psi(A)}
\label{invol}
\end{equation} and a norm 
\begin{equation}
||\Psi||:=\sup_{(g_1,\ldots,g_n)\in G^n}|\psi(g_1,\ldots,g_n)|,
\end{equation}
where $\Psi=h^{*}_{\gamma}\psi$. It can be further completed in the norm to a $C^{*}$-algebra. The Gelfand spectrum $\Abar$ of the $C^*$-algebra have an important algebraic characterization as the set $Hom(\mathcal{P},G)$ of morphisms from an appropriately defined groupoid $\mathcal{P}$ of edges in $\Sigma$ to the gauge group $G$ \cite{Thmnn}. Thanks to this characterization one can embed the space of connections $\mathcal{A}$ in the spectrum $\Abar$. Moreover it was shown that this embedding is dense \cite{AL_1}. In the sequel we will use a measure on $\Abar$ as a part of a construction of some states on a holonomy-flux algebra.

\subsection{Flux operators} 

The second step in defining the ACZ algebra is a choice of classical observables associated with the momentum part of the phase space. As explained in \cite{ACZ} functionals on the phase space linear in the momentum variable define via the Poisson bracket derivation operators on $Cyl^{\infty}$. In particular, a flux of the momentum field across an oriented face in $\Sigma$ of codimension $1$ is linear in the momentum field. The corresponding derivation operator on $Cyl^\infty$ is called a flux operator. Let us now proceed towards a precise definition of the operator.

\begin{df}
A graph $\gamma$ is adapted to a semianalytic embedded submanifold $S$ if for each edge $e_{i}\in \gamma$ one of the following conditions is true:
\begin{itemize}
\item [(i)] $e_{i}\subset \overline{S}$;
\item [(ii)] $e_{i}\cap S$ is the source of the oriented edge $e_i$; 
\item [(iii)] $e_{i} \cap S=\varnothing$.
\end{itemize}
A graph $\gamma'$ is adaptable to $S$ if there exists a graph $\gamma$ adapted to $S$ such that both $\gamma'$ and $\gamma$ share the same points in $\Sigma$.
\label{adapted}
\end{df}

Denote $D\equiv\dim\Sigma$. 
\begin{df}
A face $S$ is a connected, $(D-1)$-dimensional semianalytic embedded $C^1$ submanifold of $\Sigma$ with compact closure and a fixed orientation of its normal bundle such that every graph in $\Sigma$ is adaptable\footnote{In the literature (see eg. \cite{LOST,area}) one can find statements which seem to mean that every (semi)analytic graph is adaptable to every (semi)analytic face. This is however not true --- see Appendix \ref{counter} for a counterexample.} to $S$.   
\label{face}
\end{df}
The orientation allows us to distinguish ''upper'' $\mathcal{U}^{+}$ and ''bottom'' $\mathcal{U}^{-}$ parts of a sufficiently small neighborhood $\mathcal{U}$ open in $\Sigma$ of an arbitrary point $x\in S$. 

Let us fix a face $S$, a smearing function $f:S\to\mathfrak{g}$ of {\em compact} support on $S$ and a cylindrical function $\Psi$. There exists a graph $\gamma$ compatible with $\Psi$ and adapted to $S$. Denote by $s$ the source of an edge $e\in\gamma$. Let
\[
\theta^t(h_e(A)):=
\begin{cases}
h_e(A)\exp(tf(s)) & \text{if $e\cap S=s$ and $e$ lies 'up' the face $S$ }\\
h_e(A)\exp(-tf(s)) & \text{if $e\cap S=s$ and $e$ lies 'down' the face $S$ }\\
h_e(A) & \text{in the remaining cases}
\end{cases}
\]    
where 'up' and 'down' refer to the orientation of $S$.  
\begin{df}
A flux operator $X_{S,f}$ acts on the cylindrical function $\Psi=h^*_\gamma\psi$, where $\gamma=\{e_1,\ldots,e_n\}$ is adapted to $S$, as follows:
\[
(X_{S,f}\Psi)(A)=\frac{d}{dt} \psi(\,\theta^t(h_{e_1}(A)),\ldots,\theta^t(h_{e_n}(A))\,)\Big|_{t=0}.
\]      
\label{flux-op}
\end{df}

\subsection{ACZ algebra and  the holonomy-flux algebra \label{sec-alg}} 

The ACZ algebra (which plays the role of an algebra of elementary classical observables) is defined as a real linear space of pairs $(\Psi,X)$, where $\Psi$ is a cylindrical function and $X$ a differential operator on $Cyl^\infty$ generated by a finite number of flux operators:
\[
X=[X_{S_1,f_1},[X_{S_2,f_2},\ldots,[X_{S_{n-1},f_{n-1}},X_{S_n,f_n}], \ldots,]],
\]
where $[\cdot,\cdot]$ denotes the commutator. The algebra is equipped with a Lie bracket
\[
[(\Psi,X),(\Psi',X')]=(X'\Psi-X\Psi',[X,X'])
\]
corresponding to the Poisson structure on the phase space. Note that we constructed the algebra using $(i)$ cylindrical functions defined by means of graphs with {\em semianalytic} edges and $(ii)$ flux operators given by {\em semianalytic} faces and smearing functions of {\em compact support}, i.e. the algebra defined here is precisely the ACZ algebra considered in \cite{LOST}. Let us emphasize once again that this is not the only choice: we can consider e.g. only analytic edges and faces and also allow some smearing functions of noncompact support.  

Once the ACZ algebra is given a strictly defined procedure \cite{LOST} leads from it to the holonomy-flux algebra $\mathfrak{A}$ as a $*$-algebra of elementary quantum observables. One can think of the algebra $\mathfrak{A}$ as an algebra over $\C$ generated by elements $\{\hat{\Psi},\hat{X}_{S,f}\}$ which correspond to cylindrical functions and flux operators and commute to
\begin{equation}
[\hat{\Psi},\hat{X}_{S,f}]=i\widehat{X_{S,f}(\Psi)}. 
\label{comm}
\end{equation}
We have also an involution
\begin{equation}
\hat{\Psi}^{*}=\widehat{{\Psi}^*} \quad\hat{X}_{S,f}^{*}=\hat{X}_{S,f},
\end{equation}
where $\Psi^*$ is given by \eqref{invol}. The algebra possesses a unit given by the constant cylindrical function of value equal to $1$. It turns out that every element of $\Af$ is a finite sum of elements of the following form \cite{LOST}:
\begin{equation} 
\hat{\Psi}, \ \hat{\Psi}_1 \hat{X}_{S_{11},f_{11}},\ \hat{\Psi}_2\hat{X}_{S_{21},f_{21}}\hat{X}_{S_{22},f_{22}}, \ \ldots\ , \ 
\hat{\Psi}_{k}\hat{X}_{S_{k1},f_{k1}}\ldots \hat{X}_{S_{kk},f_{kk}},\ldots 
\label{hfa-span}
\end{equation}

\subsection{Representations of the holonomy-flux algebra}

Canonical quantization requires us to find a $*$-representation of the algebra of elementary quantum observables on a Hilbert space. Such a $*$-representation can be defined as follows \cite{AO_1}: 

\begin{df}
Let $L(\h)$ be a space of linear operators on a Hilbert space $\h$. We say that a map $\pi:\Af\rightarrow L(\h)$ is a $*$-representation of $\Af$ on the Hilbert space $\h$ if: 
\begin{enumerate}
\item there exists a dense subspace ${\cal D}$ of $\h$ such that   
\[
{\cal D}\subset\bigcap_{{a}\in\Af}[\ D(\pi({a}))\cap D(\pi({a})^*)\ ], 
\]
where $D(\pi({a}))$ denotes the domain of the operator $\pi({a})$;
\item for every ${a},{b}\in\Af$ and $\lambda\in\C$ the following conditions are satisfied on ${\cal D}$:
\begin{alignat*}{3}
\pi({a}+{b})&=\pi({a})+\pi({b}),&\qquad\pi(\lambda{a})&=\lambda\pi({a}),\\
\pi({a}{b})&=\pi({a})\pi({b}),&\qquad \pi({a}^*)&=\pi({a})^*.
\end{alignat*}  
\end{enumerate}
\label{repr-df}
\end{df}

As emphasized in the introduction we are interested in those $*$-representations of $\Af$ on a Hilbert space which are covariant with respect to gauge transformations of a theory under consideration. Since we have assumed that the bundle underlying the theory is trivial we will restrict our interest to the gauge transformations generated by diffeomorphisms of the base manifold $\Sigma$ and will consider {\em diffeomorphism covariant} $*$-representations of $\Af$. Before we will give a meaning to the notion of diffeomorphism covariance we have to define an action of the diffeomorphisms on $\Af$.      

In fact, there is a natural automorphism $\alpha_{\varphi}$ of the holonomy-flux algebra $\mathfrak{A}$ corresponding to a diffeomorphism $\varphi$ of $\Sigma$, defined on the generators of $\mathfrak{A}$ as follows \cite{Thmnn, LOST}:
\begin{equation} \label{dif_action}
 \alpha_{\varphi}\hat{\Psi}:=\widehat{\varphi^{*}\Psi},\qquad \alpha_{\varphi}\hat{X}_{S,f}:=\hat{X}_{\varphi(S),(\varphi^{-1})^{*}f},
\end{equation}
where $\varphi^{*}\Psi:=h^*_{\varphi^{-1}(\gamma)}\psi$ provided $\Psi=h^*_\gamma\psi$   and the orientation of the face $\varphi(S)$ is induced from $S$ by the diffeomorphism $\varphi$. We have also the identity
\begin{equation}
\alpha_{\varphi_{1}}\circ\alpha_{\varphi_{2}}=\alpha_{\varphi_{1}\circ\varphi_{2}},
\end{equation}
which means that the map $\varphi\mapsto\alpha_\varphi$ is a representation of the group of diffeomorphisms of $\Sigma$ on $\Af$.

Following \cite{Var} we formulate

\begin{df} \label{df_covariance}
We say that a $*$-representation $\pi:\mathfrak{A}\rightarrow\mathcal{L}(\h)$ of the holono\-my-flux algebra on a Hilbert space $\h$ is diffeomorphism covariant if on the Hilbert space there exists a unitary representation  $\varphi\mapsto U_\varphi$ of the group $\rm Diff$ of diffeomorphisms of $\Sigma$ such that
\begin{itemize}
\item [(i)]  for every $\varphi\in{\rm Diff}$  the operator $U_\varphi$ preserves the common domain $\mathcal{D}\subset\h$ (see Definition \ref{repr-df}) i.e. 
\[
U_{\varphi}\mathcal{D}\subset\mathcal{D};
\]
\item [(ii)] for every $\varphi\in{\rm Diff}$ and for every $a\in\mathfrak{A}$ 
\begin{equation}
U_{\varphi}\pi(a)U_{\varphi}^{-1}=\pi(\alpha_{\varphi}(a)).
\label{u-alpha}
\end{equation}
\end{itemize}
\end{df}

A class of  diffeomorphism covariant representations of $\Af$ can be generated via the GNS construction (see e.g. \cite{Bra_Rob,LOST}) from {\em diffeomorphism invariant states} on $\Af$. Recall that a state $\omega$ on a (unital) $*$-algebra is a linear $*$-preserving functional on the algebra valued in complex numbers which maps the unit of the algebra to $1$ and is positive, i.e. for every element $a$ of the algebra
\[
\omega(a^*a)\geq 0.
\] 
We say that a state $\omega$ on the holonomy-flux $*$-algebra $\mathfrak{A}$ is diffeomorphism invariant if for every diffeomorphism  $\varphi\in{\rm Diff}$  and every $a\in\mathfrak{A}$ 
\begin{equation}
 \omega(a)=\omega(\alpha_{\varphi}(a)).
\label{diff-inv}
\end{equation}

Given a state $\omega$ on $\Af$ the GNS construction provides us with a triple $(\h,\pi,\Omega)$ such that $(i)$ $\pi$ is a cyclic $*$-representation of $\Af$ on the Hilbert space $\h$ (in the sense of Definition \ref{repr-df}) with the cyclic vector $\Omega$ and $(ii)$ for every $a\in\Af$        
\[
 \omega(a)=\langle\Omega|\pi(a)\Omega\rangle.
\]
Assume now that \eqref{diff-inv} holds for all $a\in\Af$ and a $\varphi\in{\rm Diff}$. Then there exists a unique unitary operator $U_\varphi$ on $\h$ such that \eqref{u-alpha} is satisfied for all $a$ and $U_\varphi\Omega=\Omega$ \cite{Bra_Rob}. If $\omega$ is diffeomorphism invariant then the map $\varphi\mapsto U_\varphi$ turns out to be a unitary representation of $\rm Diff$ on $\h$ and the GNS representation $\pi$ is a diffeomorphism covariant representation of $\Af$.

\section{New diffeomorphism invariant states}

Let us recall that our goal is to show that in the case of a $U(1)$-bundle with a three- or more dimensional base manifold $\Sigma$ there are many distinct diffeomorphism invariant states on a holonomy-flux algebra provided we have constructed the algebra by allowing flux operators $\{X_{S,f}\}$ to be defined also by some smearing functions with {\em noncompact} support. In other words we are going to show existence of many distinct diffeomorphism invariant states on a modified holonomy-flux algebra. On the other hand it should be emphasized that we are not going to find all diffeomorphism invariant states on such an algebra but only some of them. Therefore we will construct them by a slight modification of formulae describing the state given by Theorem \ref{lost} called {\em the standard state} and denoted by $\omega_0$ thereafter. 

Let us first describe how we are going to define new non-standard states, then this will suggest how to modify the holonomy-flux algebra (of course, sticking to the formalism presented in Section \ref{sec-el} we will suitably modify the ACZ algebra, which will result in the desired modification of the holonomy-flux algebra). 

\subsection{The ansatz for nonstandard states \label{non-st}}

The construction of nonstandard states we are going to present follows closely the construction of a nonstandard state described in \cite{LOST}.

Let us begin by recalling formulae \cite{LOST} describing the values of the standard state on the elements \eqref{hfa-span}:
\begin{align}
& \omega_0(\hat{\Psi})=\mu_{\rm AL}(\Psi), \ \ \text{for all} \ \ \Psi\in Cyl^\infty;\label{al-m}\\
& \omega_0(\hat{\Psi}\hat{X}_{S_1,f_1}\ldots\hat{X}_{S_m,f_m})=0 \ \ \text{for all $\Psi\in Cyl^\infty$ and all $\hat{X}_{S_i,f_i}$}.\label{fl-zero}
\end{align}
In the first formula  $\mu_{\rm AL}(\Psi)$ stands for the value of an integral of $\Psi$ given by the Ashtekar-Lewandow\-ski measure $d\mu_{\rm AL}$ on the space $\Abar$ (each smooth cylindrical function, though defined on the space $\cal A$ of smooth connections, can be naturally seen as a function on $\Abar$). If a cylindrical function $\Psi=h^*_\gamma\psi$, where $\psi$ is a function on $G^n$, then      
\begin{equation}
\mu_{\rm AL}(\Psi):=\int_{\Abar}\Psi\,d\mu_{\rm AL}=\int_{G^n}\psi \,d\mu_H
\end{equation}
where $d\mu_H$ is the normalized Haar measure on $G^n$.  

In particular, the condition \eqref{fl-zero} means that
\begin{align*}
& \omega_0(\hat{X}_{S,f})=0 \ \ \text{for all $\hat{X}_{S,f}$}.
\end{align*}  
We will set the value of a new state on some flux operators to be nonzero numbers---these operators will be defined by smearing functions with noncompact support.

We choose the following simple ansatz for a nonstandard state:
\begin{equation}
\begin{aligned}
& \omega(\hat{\Psi})=\mu(\Psi), \ \ \text{for all} \ \ \Psi\in Cyl^\infty;\\
& \omega(\hat{X}_{S,f})\neq 0 \ \ \text{(in general)};\\ 
& \omega(\hat{\Psi}\hat{X}_{S_1,f_1}\ldots\hat{X}_{S_m,f_m})=\omega(\hat{\Psi})\omega(\hat{X}_{S_1,f_1})\ldots\omega(\hat{X}_{S_m,f_m})\ \ \text{for all $\hat{\Psi}$ and $\hat{X}_{S,f}$},
\end{aligned}
\label{new-st}
\end{equation}
where $\mu(\Psi)$ is given by an integral of $\Psi$ defined by a {\em normalized} measure $d\mu$ on $\Abar$. 

It is obvious that to obtain a diffeomorphism invariant state the assignment
\begin{equation}
\hat{\Psi}\mapsto\omega(\hat{\Psi}), \ \ \ \hat{X}_{S,f}\mapsto\omega(\hat{X}_{S,f})
\label{assign}
\end{equation}
has to be diffeomorphism invariant. Therefore the measure $d\mu$ on $\Abar$ used to define the value $\omega(\hat{\Psi})$ also has to be diffeomorphism invariant. To make the value $\omega(\hat{X}_{S,f})$ invariant with respect to diffeomorphisms we will set it to be equal to a limit of the smearing function $f:S\to u(1)\cong\R$ at a point distinguished in a diffeomorphism covariant manner. The idea is to remove an isolated point $x_0$ from a face $\tilde{S}$. Then the assignment 
\begin{equation}
(S,f)\mapsto \lim_{x\to x_0}f,
\label{assign-2}
\end{equation}
where $S:=\tilde{S}\setminus\{x_0\}$, is diffeomorphism invariant (provided the limit exists). Note that in order to get a nonzero limit the support of the function $f$ cannot be compact.  

There is however a difficulty here: although there is nothing wrong with \eqref{assign-2} the resulting assignment $\hat{X}_{S,f}\mapsto \lim_{x\to x_0}f$ is ambiguous: the identity $\hat{X}_{S,f}=\hat{X}_{-S,-f}$, where $-S$ is obtained from $S$ by the change of its orientation, gives
\[
\lim_{x\to x_0}f \mapsfrom\hat{X}_{S,f}=\hat{X}_{-S,-f}\mapsto -\lim_{x\to x_0}f.
\]
This means that to get rid of the ambiguity we have to distinguish one of the orientations of the face. It is clear that usually we cannot distinguish any orientation of a face in a {\em diffeomorphism invariant} manner which is necessary in our construction. Fortunately, we can do this if the face is e.g. a cone---we can assign to a cone {\em with the vertex removed} and with a fixed orientation the value of the limit \eqref{assign-2} setting $x_0$ to be the vertex of the cone. This assignment is unambiguous, diffeomorphism invariant and nontrivial (i.e. it gives nonzero value, in general).      

Thus to a flux operator based on a cone or a face of analogous properties we will assign the value of the limit \eqref{assign-2}, and in remaining cases we will set $\omega(\hat{X}_{S,f})=0$. 

The idea of constructing new nonstandard states seems to be rather simple. However, there is much more to be done: one has to show that the ansatz can be extended by linearity to a state on the algebra. The task is non-trivial not only because one has to prove the positivity of the resulting map---note that, although the elements \eqref{hfa-span} used in the formulae \eqref{new-st} span (any modification of) the holonomy-flux algebra, they {\em do not} form a basis of the algebra, therefore one also has to prove that the ansatz respect the linear dependence between the elements \eqref{hfa-span}. 

\subsection{Some properties of the ansatz}

As just mentioned the elements \eqref{hfa-span} span any holonomy-flux algebra irrespectively of details concerning the principle bundle, its structure group, the class of graphs, faces and smearing functions underlying a particular construction. Now we are going to describe some properties of the ansatz which are also independent on these details.

\begin{lm}
Suppose that $\omega$ is a linear functional on a holonomy-flux algebra $\mathfrak{A}$ which satisfies \eqref{new-st}. Then $\omega$ is real, i.e. $\omega(a^{*})=\omega(a)^{*}$ for all $a\in\Af$ if and only if $\omega(\hat{\Psi})=\mu_{\rm AL}(\Psi)$ for all cylindrical functions and $\omega(\hat{X}_{S,f})=\omega(\hat{X}_{S,f})^{*}$ for all flux operators.
\label{lm-real}
\end{lm}
\begin{proof}
Assume first that $\omega$ is real. Since $\hat{X}_{S,f}^{*}=\hat{X}_{S,f}$  the condition $\omega(\hat{X}_{S,f})=\omega(\hat{X}_{S,f})^{*}$ follows immediately. 

Consider now an element $\hat{\Psi}\hat{X}_{S,f}+\hat{X}^{*}_{S,f}\hat{\Psi}^{*}$. It is preserved by the involution, hence $\omega(\hat{\Psi}\hat{X}_{S,f}+\hat{X}^{*}_{S,f}\hat{\Psi}^{*})\in \mathbb{R}$. On the other hand using \eqref{new-st} and \eqref{comm} we get
\begin{equation*}
\omega(\Psi\hat{X}_{S,f}+\hat{X}^{*}_{S,f}\Psi^{*})=\omega(\Psi+\Psi^{*})\omega(\hat{X}_{S,f})-i\omega(\widehat{X_{S,f}(\Psi)})\in \mathbb{R}.
\end{equation*}
Thus for every $\Psi\in Cyl^\infty$ and for every flux operator 
\begin{equation*}
0=\omega(\widehat{X_{S,f}(\Psi)})=\mu(X_{S,f}(\Psi)).
\end{equation*}
It was proven in \cite{AO_1,AO_2} that in such a case $\mu=\mu_{AL}$ (note that we assumed that the measure $d\mu$ is normalized). 

Let us assume now that $\omega(\hat{\Psi})=\mu_{AL}(\Psi)$ and $\omega(\hat{X}_{S,f})=\omega(\hat{X}_{S,f})^{*}$. It is clear that $\omega$ is real on the subalgebra generated by cylindrical functions. Thus what remains to be done is to show the reality of $\omega$ on the elements
\begin{equation*}
\hat{\Psi}\hat{X}_{S_{1},f_{1}}\ldots \hat{X}_{S_{k},f_{k}}.
\end{equation*}
In the remaining part of the proof in order to make the notation more transparent we will denote $\hat{X}_k\equiv \hat{X}_{S_k,f_k}$.  

We begin with proving that
\begin{equation} \label{positivity_eq_1}
 \omega(\hat{X}_{1}\ldots\hat{X}_{k}\hat{\Psi}\hat{X}_{k+1}\ldots\hat{X}_{n})=\omega(\hat{\Psi}\hat{X}_{1}\ldots\hat{X}_{n}).
\end{equation}
Indeed,
\begin{multline*}
 \omega(\hat{X}_{1}\ldots\hat{X}_{k}\hat{\Psi}\hat{X}_{k+1}\ldots\hat{X}_{n})=\omega(\hat{X}_{1}\ldots(-i\widehat{{X}_{k}(\Psi)})\hat{X}_{k+1}\ldots\hat{X}_{n})+\\+\omega(\hat{X}_{1}\ldots\hat{X}_{k-1}\hat{\Psi}\hat{X}_{k}\ldots\hat{X}_{n})=\omega(\hat{\Psi}\hat{X}_{1}\ldots\hat{X}_{n})+\omega(\text{a sum of other terms}) 
\end{multline*}
where each of the other terms is of the form $\widehat {X'_{0}(\Phi)}\hat{X}'_{1}\ldots\hat{X}'_{m}$, where $X'_{0}(\Phi)\in Cyl^\infty$. But 
\[
\omega(\widehat{X_i(\Phi)})=\mu_{\rm AL}(X_i(\Psi))=0,
\]
hence taking into account the last property of \eqref{new-st} 
\[
\omega(\text{a sum of other terms})=0
\]
and  \eqref{positivity_eq_1} follows.

Applying \eqref{positivity_eq_1} we obtain:
\begin{multline*}
\omega((\hat{\Psi}\hat{X}_{1}\ldots \hat{X}_{k})^{*})=\omega(\hat{X}^{*}_{k}\ldots\hat{X}_{1}^{*}\hat{\Psi}^{*})=\omega(\hat{X}_{k}\ldots\hat{X}_{1}\hat{\Psi}^{*})=\omega(\hat{\Psi}^{*}\hat{X}_{k}\ldots \hat{X}_{1})
=\\=\omega(\hat{\Psi}^{*})\omega(\hat{X}_{1})\ldots \omega(\hat{X}_{k})=\omega(\hat{\Psi})^{*}\omega(\hat{X}_{1})^{*}\ldots \omega(\hat{X}_{k})^{*}= \\
=\omega(\hat{\Psi}\hat{X}_{1}\ldots \hat{X}_{k})^{*}. 
\end{multline*}
which completes the proof.
\end{proof}

\begin{lm} \label{lem_positivity}
 Let $\omega$ be a real linear functional on a holonomy-flux algebra. If $\omega$ satisfies  \eqref{new-st} then it is positive. Moreover, it is a state on the algebra.
\end{lm}
\begin{proof}
 It follows from the previous lemma that $\mu=\mu_{\rm AL}$. Our task now is to prove that $\omega(a^{*}a)\geq 0$ for every element $a$ of the algebra. Such an element can be expressed as a finite sum of elements of the form \eqref{hfa-span}:
\begin{equation*}
 a=\sum_{i}\hat{\Psi}_{i}\hat{X}_{1i}\ldots\hat{X}_{k_i i}
\end{equation*}
(note that we again denoted flux operators without indicating faces and smearing functions explicitely). Thus we get:
\begin{multline*}
\omega(a^{*}a) = \omega((\sum_{j}\hat{X}^*_{k_j j}\ldots\hat{X}^*_{1j}\hat{\Psi}_{j}^{*})(\sum_{i}\hat{\Psi}_{i}\hat{X}_{1i}\ldots\hat{X}_{k_i i}))=
\\
=\omega(\sum_{ij}\hat{X}_{k_j j}\ldots\hat{X}_{1j}\hat{\Psi}^*_{j}\hat{\Psi}_{i}\hat{X}_{1i}\ldots\hat{X}_{k_i i})=\\
= \sum_{ij}\omega(\hat{\Psi}^{*}_{j}\hat{\Psi}_{i}\hat{X}_{k_j j}\ldots\hat{X}_{1 j}\hat{X}_{1i}\ldots\hat{X}_{k_i i})=
\\
= \sum_{ij}\mu_{\rm AL}(\hat{\Psi}^{*}_{j}\hat{\Psi}_{i})\omega(\hat{X}_{k_j j})\ldots\omega(\hat{X}_{1 j})\omega(\hat{X}_{1i})\ldots\omega(\hat{X}_{k_i i})
\end{multline*}
where we used \eqref{positivity_eq_1} and the last property of \eqref{new-st}. Denote
\[
\beta_i=\omega(\hat{X}_{1i})\ldots\omega(\hat{X}_{k_i i}).
\]
Because of reality of $\omega$ we have ${\beta}^*_i=\beta_i$. This property and assumed reality of $\omega$ gives us  
\begin{multline*}
\sum_{ij}\mu_{\rm AL}({\Psi}^{*}_{j}{\Psi_{i}})\omega(\hat{X}_{1j})\ldots\omega(\hat{X}_{k_j j})\omega(\hat{X}_{1i})\ldots\omega(\hat{X}_{k_i i})=\\=\sum_{ij}\mu_{\rm AL}({\Psi}^{*}_{j}{\Psi_{i}}){\beta}^*_j\beta_i=\sum_{ij}\mu_{\rm AL}({\beta^*_j}{\Psi}^{*}_{j}\beta_i{\Psi_{i}})=\mu_{\rm AL}({\Phi}^{*}{\Phi})\geq 0,
\end{multline*}
where ${\Phi}\equiv\sum_{i}\beta_{i}{\Psi}_{i}\in Cyl^\infty$.

Thus we proved positivity of $\omega$. Let us recall that the unit $\hat{1}$ of any holonomy-flux algebra is given by a constant cylindrical function of value equal to $1$. Since the Ashtekar-Lewandowski measure is normalized we have 
\[
\omega(\hat{1})=\mu_{\rm AL}(1)=1.
\] 
This together with positivity of $\omega$ means that $\omega$ is a state on the algebra.  
\end{proof}

\subsection{Modified algebras}

Consider the principle bundle $P=\R^D\times U(1)$ with $D\geq3$. We will treat $\R^D$ as a {\em real-analytic} manifold. Consequently, we will build the ACZ algebra from cylindrical functions and flux operators based on, respectively, graphs of {\em analytic} edges and {\em analytic} faces. 

Let us make precise the notion of ``a point removed from a face'' used in the description of the idea of constructing new nonstandard states:
\begin{df}
We say that $x\in\partial S:=\overline{S}\setminus S$ \footnote{It follows from the definition of $S$ as an embedded submanifold that $S$ coincides with its own interior with respect to the (induced) topology of $\overline{S}$.} is an isolated boundary point of $S$ if there exists a neighborhood $U$, open in $\R^D$, such that $\partial S\cap U=\{x\}$. 
\end{df}

The set of all isolated boundary points of $S$ will be denoted by $I_{\partial S}$. It is clear that every such a set has the ''diffeomorphic covariance'' property, that is, for any diffeomorphism $\varphi$ of $\R^D$:
\begin{equation}
I_{\partial\varphi(S)}=\varphi(I_{\partial S}).
\label{I_cov}
\end{equation}
In the sequel we will often make use of a completion $\tilde{S}$ defined as a sum of a face $S$ and its isolated boundary points:   
\begin{equation} \label{completion}
\tilde{S}:=S\cup I_{\partial S}.
\end{equation}
The completion $\tilde{S}$ may possess isolated boundary points\footnote{Let $S'$ be a face with no isolated boundary points and let $(x_n)_{n=1,2,3\ldots}$ be a sequence of pairwise distinct points of $S'$ convergent to $x_0\in S'$. Then $x_0$ is an isolated boundary point of the completion $\tilde{S}$ of a face $S:=S'\setminus\{\ x_n \ | \ n=0,1,2,3\ldots\ \}$.}.

Now we are able to describe precisely a class of smearing functions suitable for our construction.
\begin{df} \label{def_fun}
The smearing function $f$ on a face $S$, $f:S\rightarrow \mathbb{R}$ is allowable if there exists a compactly supported continuous function $\tilde{f}$ on $\tilde{S}$ such that $\tilde{f}$ restricted to $S$ coincides with $f$.
\end{df}

Now we are at a point to define the modified algebras:
\begin{df}
$(i)$ The modified ACZ algebra is an algebra of elementary classical observables constructed according to the description in Section \ref{sec-alg} from cylindrical functions based on graphs with analytic edges and flux operators based on analytic faces and smearing functions given by Definition \ref{def_fun}. $(ii)$ The $*$-algebra of quantum observables generated by the modified ACZ algebra will be called  modified holonomy-flux algebra and denoted by $\tilde{\Af}$.  
\end{df}

\subsection{Cone-like faces \label{cone-like}}

As described above new nonstandard states will assign nonzero values to flux operators based on cones and other similar faces. In the present subsection we shall give a precise definition of these faces and describe their properties.

\begin{df}
A generalized cone of type $d\geq 1$ (a $d$-cone) is an image under an analytic diffeomorphism of $\R^D$ of the set  
\begin{equation}
\{\ (x^1,\ldots,x^D)\in\R^D\ | \ (x^{1})^{2}+\ldots+(x^{D-1})^{2}-(x^{D})^{2d}=0 \ \ \text{and} \ \ 0 < x^{D} < H\ \}, 
\label{cone}
\end{equation}
where $(x^{1},\ldots,x^{D})$ are canonical coordinates on $\R^D$ and $H>0$.
\end{df} 

It is clear that the set of isolated boundary points of any $d$-cone consists of one point which will be called the {\em vertex} of the cone.

Every $d$-cone is an analytic embedded submanifold of $\R^D$. For every $d$-cone $C^d$ we can find an analytic embedded connected $(D-1)$-dimensional submanifold $C^d_M$  such that $(i)$ $C^d\subset C^d_M$ and $(ii)$ if $C^d_M$ is contained in an analytic embedded connected $(D-1)$-dimensional submanifold $C$ then $C^d_M=C$. We will call the manifold $C^d_M$ {\em the maximal analytic extension of} $C^d$. $C^d_M$ can be shown to be  an image of 
\begin{equation}
\{\ (x^1,\ldots,x^D)\in\R^D\ | \ (x^{1})^{2}+\ldots+(x^{D-1})^{2}-(x^{D})^{2d}=0 \ \ \text{and} \ \ 0 < x^{D}\ \}
\label{max}
\end{equation}
under an appropriate analytic diffeomorphism.

We do not guarantee that  a cone of type $d>1$ cannot be at the same time a cone of type $d'>1$ with $d'\neq d$ --- we cannot exclude that there exists an analytic diffeomorphism which maps a set \eqref{cone} with $d>1$ into a set \eqref{cone} with $d'>1$ distinct from $d$. But if a cone is of type $1$ then it cannot be a cone of type $d>1$ --- this is excluded by Lemma \ref{ineq}.

\subsubsection{Orientation of $d$-cones}

The following lemma states a property of $d$-cones, which is crucial for the construction  of nonstandard states: 
\begin{lm}
Let $C^{d}$ be a $d$-cone with a fixed orientation of its normal bundle. Then
every analytic diffeomorphism $\R^D$ which maps the cone onto itself preserves the orientation.
\label{orient}
\end{lm}

To give a proof of the lemma we need some preparations. Consider a $C^1$-class curve 
\begin{equation}
]-a,a[\ni\tau\mapsto\kappa(\tau)=(x^1(\tau),\ldots,x^D(\tau))\in\R^D,\ \ a>0.
\label{kappa-cor}
\end{equation}
We will say that the curve $\kappa$ is ``ingoing'' to a cone $C^d$ through its vertex $v$ if 
\begin{enumerate}
\item $\kappa(]0,a[)\subset C^d$,
\item $\kappa(0)$ coincides with the vertex $v$ of $C^d$. 
\end{enumerate}
We will denote by $\dot{\kappa}$ the vector tangent to $\kappa$ at $v$. The $D$-th component of $\dot{\kappa}$, $\dot{\kappa}^D$,  refers to the decomposition of the vector in the basis given by the canonical coordinates $(x^1,\ldots,x^D)$ on $\R^D$ and is given by
\[
\dot{\kappa}^D=\frac{dx^D(\tau)}{d\tau}\Big|_{\tau=0}.
\]
Given a diffeomorphism $\varphi$ on $\R^D$ preserving the vertex $v$, $\varphi'$ will denote the tangent map from $T_v\R^D$ onto itself given by the diffeomorphism.     

Consider now the cone $C^d$ given by \eqref{cone} and a curve $\kappa$ ingoing to the cone through its vertex $v=(0,\ldots,0)$. Then
\begin{equation}
\dot{\kappa}^D=\lim_{\tau\to 0^+}\frac{x^D(\tau)-x^D(0)}{\tau}=\lim_{\tau\to 0^+}\frac{x^D(\tau)}{\tau}\geq 0,
\label{Dth}
\end{equation}
because $x^D(\tau)>0$ for $\tau>0$, which follows from the fact that $\kappa$ is ingoing and from the definition \eqref{cone} of $C^d$.
     
Suppose now that a diffeomorphism $\varphi$ preserves the cone $C^d$. Consequently, $\varphi$ maps the vertex $v$ onto itself. It also maps any ingoing curve to ingoing one, hence the property \eqref{Dth} is preserved by the diffeomorphism. This conclusion can be stated more precisely as follows:
\begin{cor}
Let $C^d$ be given by \eqref{cone} and let $\kappa$ be a curve ingoing to  $C^d$ through its vertex $v$. If $\varphi$ is a diffeomorphism preserving the cone then
\[
(\varphi'\dot{\kappa})^D\geq 0.
\]          
\label{in+}
\end{cor}

In fact, a bit stronger statement holds:
\begin{lm} \label{lm_outgoing}
Let $C^d$ be given by \eqref{cone} and let $\kappa$ be a curve ingoing to  $C^d$ through its vertex $v$ such that $\dot{\kappa}^D>0$. If $\varphi$ is a diffeomorphism preserving the cone then
\[
(\varphi'\dot{\kappa})^D>0.
\]
\label{dk>0}
\end{lm}

\begin{proof}
Consider an ingoing curve $\kappa$. Assume that $\dot{\kappa}^D>0$ and $(\varphi'\dot{\kappa})^D=0$. This means that the vector $\varphi'\dot{\kappa}$ is tangent to the hyperplane $x^D=0$. Hence we can use a rotation $R$ around the $x^D$-axis to map $\varphi'\dot{\kappa}$ into $-\varphi'\dot{\kappa}$. Then we can use $\varphi^{\prime -1}$ to map  $-\varphi'\dot{\kappa}$ into $-\dot{\kappa}$. Since the rotation $R$ is an analytic diffeomorphism preserving the cone $C^d$ we conclude that the diffeomorphism $\phi:=\varphi^{-1}\circ R\circ\varphi$ preserves the cone and the corresponding tangent map $\phi'$ maps $\dot{\kappa}$ into $-\dot{\kappa}$. We arrive at the conclusion that $(\phi'\dot{\kappa})^D<0$. This however contradicts Corollary \ref{in+}. 
\end{proof}

Consider now the vector $\vec{e}_D=(0,\ldots,0,1)$ based at the vertex $v$ of the cone $C^d$ (given by \eqref{cone}) and a diffeomorphism $\varphi$ which preserves the cone. 

If $d>1$ then there exist a curve $\kappa$ ingoing to the cone through its vertex such that $\dot{\kappa}=\vec{e}_D$. By virtue of Lemma \ref{dk>0} for any diffeomorphism $\varphi$ preserving the cone $C^d$
\begin{equation}
(\varphi' \vec{e}_D)^D>0.
\label{e_D>0}
\end{equation}

Assume now that $d=1$. Then 
\[
\vec{e}_D=\Big(\frac{1}{2},0,\ldots,0,\frac{1}{2}\Big)+\Big(-\frac{1}{2},0,\ldots,0,\frac{1}{2}\Big).
\]  
It is easy to find two curves ingoing to the cone such that they generate the two vectors at the r.h.s. of the equation above. Again by virtue of Lemma \ref{dk>0} for any diffeomorphism $\varphi$ preserving the cone $C^1$
\[
\Big(\varphi'\Big(\pm\frac{1}{2},0,\ldots,0,\frac{1}{2}\Big)\Big)^D>0
\]
and consequently \eqref{e_D>0} holds also in the case $d=1$.
 
Thus we arrived at the following conclusion:
\begin{cor}
Assume that a diffeomorphism $\varphi$ preserves the cone $C^d$ given by \eqref{cone} and  $\vec{e}_D=(0,\ldots,0,1)$ is based at the vertex of $C^d$. Then  the $D$-th component of $\varphi'\vec{e}_D$ is positive.     
\label{comp}
\end{cor}         

\begin{proof}[Proof of Lemma \ref{orient}]
It is not difficult to realize that since the lemma concerns analytic diffeomorphisms its assertion holds if and only if it holds true for $C^d_M$ being the maximal analytic extension of the cone $C^d$ given by \eqref{cone}. This extension is unbounded and its closure divides the space $\mathbb{R}^{D}$ in two disjoint open subsets: the ''internal'' one denoted by $U_+$ (the $x^D$ coordinate of every point in $U_+$ is positive) and the ''external'' one $U_-$. Obviously,  
\begin{equation*}
 \mathbb{R}^{D}=U_+\cup U_-\cup C^d_M\cup \{ v \},
\end{equation*} 
where $v$ is the vertex of $C^d_M$.
 
Consider now the constant vector field $Y=(0,\ldots,0,-1)$ on $\R^D$. This vector field restricted to $C^d_M$ defines an orientation of the cone. Denote by $\chi^0$ the integral curve of $Y$ which runs through the vertex $v$---$\chi^0$ coincides with the $x^D$-axis with the opposite orientation. For every integral curve $\chi$ of the vector field we define 
\[
 \chi_{\pm}:=\chi\cap U_{\pm}.
\]

Let $\varphi$ be an analytic diffeomorphism preserving $C^d_M$. Suppose that $\varphi$ changes the orientation of the cone. Then for arbitrary integral curve $\chi\neq \chi^0$ of the vector field $Y$  
\begin{equation}
\begin{aligned}
 & \varphi(\chi_\pm)\cap U_{\pm}=\varnothing \\
&\varphi(\chi_\pm)\cap U_{\mp}=\varphi(\chi_\pm).
\end{aligned}
\label{chi-U}
\end{equation}
But because of continuity of both the diffeomorphism and the vector field the above statements have to be true also for $\chi^0$. 

Note now that the vector tangent to $\chi^0$ at the vertex $v$ is $-\vec{e}_D$. Since $\varphi(\chi^0)$ satisfies \eqref{chi-U} the $D$-th component of 
\[
\varphi'(-\vec{e}_D)=-\varphi'\vec{e}_D
\]       
has to be {\em non-negative}. This however contradicts Corollary \ref{comp} which implies that every $\varphi$ preserving $C^d_M$ maps the vector $-\vec{e}_D$ into one of a {\em negative} $D$-th component. 

Thus we conclude that if $\varphi$ preserves $C^d_M$ then it cannot change its orientation.       
\end{proof}

Among faces used to define flux operators are those which are subsets of $d$-cones. It is now clear that Lemma \ref{orient} implies the following conclusion:
\begin{cor}
If a face $S$ is a subset of a $d$-cone then every analytic diffeomorphism of $\R^D$ which preserves $S$ preserves also the orientation of $S$.
\label{cor-or}     
\end{cor}

Now let us distinguish one of the orientations of a $d$-cone $C^{d}$ as \emph{external}: we mean that a $d$-cone face $C^{d}$ is externally oriented if the "up" side of $C^{d}$ is "pointed out" by its vertex. More precisely: the external orientation of the $d$-cone \eqref{cone} is that given by the constant vector field $Y=(0,\ldots,0,-1)$ on $\R^D$ restricted to the cone. The externally oriented $C^d$ is the image under a diffeomorphism $\varphi$ of the externally oriented $d$-cone \eqref{cone}.

\subsubsection{Some auxiliary facts}

In this subsection  we will present and justify some auxiliary statements of a rather technical character. We will also introduce a notion of almost equal flux operators. All these  will be used in the next section to state and prove  the main theorem of the paper.

\begin{lm}
Let a face $S$ be a subset of a $d$-cone $C^d$. Then the set $I_{\partial S}\cap C^d$ is nowhere dense in $C^d$. 
\label{Ips-nd}
\end{lm}

\begin{proof}
Let $x\in I_{\partial S}\cap C^d$. We know that there exists a nonempty set $U_x$ open in $\R^D$ such that $U_x\cap \partial S=\{x\}$. Hence $U_x\cap I_{\partial S}=\{x\}$ and $C^d\cap U_x\cap I_{\partial S}=\{x\}$. But $C^d\cap U_x\equiv V_x$ is nonempty and open in $C^d$. Then
\[
V:=\bigcup_{x\in I_{\partial S}\cap C^d}V_x
\]
is nonempty and open in $C^d$.

Consider now a nonempty set $A\subset C^d$ open in $C^d$. Assume that $A\cap V$ is nonempty. Then there exists $x\in I_{\partial S}\cap C^d$ such that $A\cap V_x$ is nonempty. It is excluded that $A\cap V_x=\{x\}$ since $\{x\}$ is not open in $C^d$ which follows from the fact that $C^d$ is an embedded submanifold of $\R^D$. Hence $A\cap(V_x\setminus\{x\})$ is nonempty and open and does not contain any points of $I_{\partial S}\cap C^d$.

If $A\cap V$ is empty then $A$ again does not contain any points of $I_{\partial S}\cap C^d$.  

We just proven that every nonempty open subset $A$ of $C^d$ contains a nonempty open subset $B$ such that $B\cap (I_{\partial S}\cap C^d)=\varnothing$. This means that $I_{\partial S}\cap C^d$ is nowhere dense subset of $C^d$ \cite{top}.
\end{proof}

\begin{lm} \label{lem-X-ext}
Let $X_{S,f}\in\tilde{\Af}$ be a flux operator based on a face $S$ being a subset of a $d$-cone. Then there exists a $d$-cone $C^d$ and a unique allowable smearing function $f':C^d\mapsto\R$ such that $S\subset C^d$ and
\begin{equation}
f'(x)= 
\begin{cases}
f(x) & \text{if $x\in S$}\\
0 &  \text{if $x\in C^d\setminus\tilde{S}$}
\end{cases}.
\label{f'}
\end{equation}
\label{X-ext}     
\end{lm}  

\begin{proof}
Without loss of generality we can prove the lemma assuming that $S$ is a subset of a $d$-cone given by \eqref{cone}. As the cone $C^d$ let us take a $d$-cone given by \eqref{cone} of such a height $H$ that the closure $\overline{S}$ either does not intersect the boundary of $C^d$ or intersect the boundary only at the vertex of $C^d$.         

Note now that $\tilde{S}\subset \tilde{C}^d$. Indeed, it follows from the above assumption about $C^d$ that $\overline{S}\subset \tilde{C}^d$. On the other hand $\tilde{S}\subset\overline{S}$.     

According to Definition \ref{def_fun} there exists a compactly supported continuous function $\tilde{f}$ on $\tilde{S}$ such that $f=\tilde{f}|_S$. Therefore for every $x_0\in\partial\tilde{S}$ (where $\partial\tilde{S}$ is the boundary of $\tilde{S}$ as a subset of $\tilde{C}^d$)
\[
\lim_{x\to x_0}\tilde{f}(x)=0.
\]     
This means that the function $\tilde{f}':\tilde{C}^d\to\R$ defined as
\begin{equation}
\tilde{f}'(x):=
\begin{cases}
\tilde{f}(x) & \text{if $x\in\tilde{S}$}\\
0 & \text{otherwise}
\end{cases}
\label{f'f}
\end{equation}
is continuous and of compact support in $\tilde{C}^d$. Hence $f':=\tilde{f}'|_{C^d}$ is an allowable smearing function on $C^d$. Moreover it satisfies the condition \eqref{f'}. In this way we showed existence of $f'$.

In order to show uniqueness of $f'$ it is enough to note that $f'$ is continuous on $C^d$ and the condition \eqref{f'} defines its values everywhere on $C^d$ except the nowhere dense set $I_{\partial S}\cap C^d$.       
\end{proof}

Note that $X_{S,f}\neq X_{C^d,f'}$ because there may exist isolated boundary points of $S$ belonging to $C^d$ for which the values of $f'$ are non-zero. However, as stated by Lemma \ref{Ips-nd} the set $I_{\partial S}\cap C^d$ is quite ``small''. Therefore we will say that $\hat{X}_{S,f}$ and $\hat{X}_{C^d,f'}$ such that $(S,f)$ and $(C^d,f')$ are related as in Lemma \ref{lem-X-ext} are {\em almost equal} and will write        
\[
\hat{X}_{S,f}\approx \hat{X}_{C^d,f'}.
\]

Now we can draw the following conclusion:

\begin{cor}
Assume that
\[
\hat{X}_{S,f}\approx \hat{X}_{C_1^d,f'_1} \ \ \ \text{and} \ \ \ \hat{X}_{S,f}\approx \hat{X}_{C_2^{d'},f'_2}.
\]
Then either $d=d'=1$ or both $d,d'$ are greater than $1$. Moreover, the cones share the common vertex $v$ and  
\[
\tilde{f}'_1(v)=\tilde{f}'_2(v).
\]  
\label{cc-dd'}
\end{cor}
\begin{proof}
Because $S\subset C_1^d$ and $S\subset C_2^{d'}$ and because the faces are analytic all of them are contained in the maximal analytic extension of $C^d_1$. Therefore either $d=d'=1$ or both $d,d'$ are greater than $1$ and the cones share the same vertex $v$. 

Moreover, there  exist a neighborhood $U$ of $v$ in $\R^D$ such that $U\cap {C}_1^d=U\cap {C}_2^{d'}$. By virtue of Lemmas \ref{lem-X-ext} and \ref{Ips-nd} the functions ${f}'_1$ and ${f}'_2$ coincide on $U\cap {C}_1^d$ hence $\tilde{f}'_1(v)=\tilde{f}'_2(v)$. 
\end{proof}

\begin{lm}
Let $c_M$ be the closure of the maximal analytic extension of a $d$-cone $C^d$. Suppose that among analytical faces $\{S_1,\ldots,S_k\}$ there is no one which is a subset of $c_M$. Then
\begin{equation}
\bigcup_{i=1}^k C^d\cap S_i
\label{ndense}
\end{equation}
is a nowhere dense subset of $C^d$. 
\label{lm-dense}
\end{lm}

\begin{proof}
Suppose first that, given $S_i$, the set $c_M\cap S_i$ is not a nowhere dense subset of $c_M$. Then using the continuity of the face $S_i$ one can show that $c_M\cap S_i$ contains an open subset of $c_M$, hence $S_i\subset c_M$ by the analyticity of both faces which contradicts the assumptions of the lemma.

Thus for any $i$ the set $c_M\cap S_i$ is a nowhere dense subset of $c_M$. A finite union of nowhere dense subsets is nowhere dense again \cite{top}, thus 
\[
\bigcup_{i=1}^k {c_M\cap S_i} 
\]
is nowhere dense. $C^d$ is an embedded submanifold of $\R^D$ and therefore it is an open subset of $c_M$. Hence \eqref{ndense} is nowhere dense in $C^d$. 
\end{proof}

Consider now a cone $C^d$ and a finite set of faces $\{S_1,\ldots,S_k\}$ contained in $C^d$ and a finite set of faces $\{S'_1,\ldots,S'_{k'}\}$ such that no one is contained in the closure of the maximal analytic extention of $C^d$. Then it follows immediately from Lemmas \ref{Ips-nd} and \ref{lm-dense} that the set
\begin{equation}
\Big(\bigcup_{i=1}^{k} C^d\cap I_{\partial S_i} \Big)\cup \Big(\bigcup_{j=1}^{k'} C^d\cap S'_j\Big)
\label{sss}
\end{equation}
is nowhere dense in $C^d$. Hence we have the following corollary: 
\begin{cor}
The set
\begin{equation}
\check{C}^d:=C^d\setminus \ \text{\rm the set \eqref{sss}} 
\label{dens-C}
\end{equation}
is dense in $C^d$. 
\label{c-dens-C}
\end{cor}

\subsection{Main theorem}

Since now we will treat each flux operator as based on either $(i)$ a face being a subset of a $d$-cone with an orientation consistent with the {\em external} orientation of the cone or $(ii)$ a face which is not contained in any $d$-cone (if the orientation of a face contained in a $d$-cone is not consistent with the external orientation of the cone then we change it together with the sign of the smearing function on the face). Using this convention we state the main theorem of this paper:

\begin{thm}
The following assignment $\omega$ of a complex number to elements of the modified holonomy-flux algebra $\tilde{\mathfrak{A}}$
\begin{align}
&\omega(\hat{\Psi})=\mu_{AL}(\Psi), \nonumber\\
&\omega(\hat{X}_{S,f})=
\begin{cases}
\chi(d)\,\tilde{f}(v) &\text{if $S$ is a $d$-cone of vertex $v$ }\\
\omega(\hat{X}_{C,f'}) & \text{if $S$ is a subset of a $d$-cone $C$ and  $\hat{X}_{S,f}\approx \hat{X}_{C,f'}$ }\\
0 & \text{otherwise}
\end{cases},
\label{tw1_f2}
\\
& \omega(\hat{\Psi}\hat{X}_{S_1,f_1}\ldots\hat{X}_{S_m,f_m})=\omega(\hat{\Psi})\omega(\hat{X}_{S_1,f_1})\ldots\omega(\hat{X}_{S_m,f_m})\ \ \text{for all $\hat{\Psi}$ and $\hat{X}_{S,f}$},\nonumber
\end{align}
where $\chi:[1,\infty[\to\mathbb{R}$ is a function constant on $]1,\infty[$, extends by linearity to a diffeomorphism invariant state on the algebra.
\label{main}
\end{thm}
\begin{rem}
The value $\tilde{f}(v)$ is equal to $\lim_{x\to v}f(x)$, thus the assignment above agrees with the preliminary description given in Subsection \ref{non-st}.   
\end{rem}

Note that the assignment \eqref{tw1_f2} is defined also for a product of flux operators. Indeed,  
\[
\omega(\hat{X}_{S_1,f_1}\ldots\hat{X}_{S_m,f_m})=\omega(\hat{1}\hat{X}_{S_1,f_1}\ldots\hat{X}_{S_m,f_m})=\omega(\hat{X}_{S_1,f_1})\ldots\omega(\hat{X}_{S_m,f_m}),
\]
where $\hat{1}$ is the unit of $\tilde{\mathfrak{A}}$ given by the constant cylindrical function $\Psi_0$ of value $1$ for which $\omega(\hat{1})=\mu_{\rm AL}(\Psi_0)=1$.

\begin{proof}[Proof of Theorem \ref{main}.] 
For the sake of clarity let us divide the proof into some steps. In each step we will justify a statement:
\begin{description}
\item[\hspace{10pt} Step 1:] The assignment \eqref{tw1_f2} is well defined.
\item[\hspace{10pt} Step 2:] If the assignment defines a linear functional on $\tilde{\mathfrak{A}}$ then the functional is a state on the algebra.  
\item[\hspace{10pt} Step 3:] If the assignment defines a linear functional then the functional is diffeomorphism invariant. 
\item[\hspace{10pt} Step 4:] The assignment defines a linear functional on $\tilde{\mathfrak{A}}$.
\end{description}
Once we will justify all the statements the proof of the theorem will be complete.\smallskip

{\bf Step 1.} Let us first check whether the assignment \eqref{tw1_f2} is unambiguously defined --- a potential source of an ambiguity is twofold: $(i)$ to define the assignment we use the type $d$ of a cone, but we are not sure whether the type can be uniquely ascribed to the cone unless $d=1$  and $(ii)$ the relation  $\hat{X}_{S,f}\approx \hat{X}_{C,f'}$ is not a one-to-one relation.  Assume then that $S$ is at the same time a cone of type $d>1$ and a cone of type $d'>1$ with $d'\neq d$. Then it follows from the properties of the function $\chi$ that the value $\omega(\hat{X}_{S,f})$ is well defined. Regarding $(ii)$ assume now that $S$ is a subset of a $d$-cone $C$. Then there are many cones containing $S$ and therefore there are many operators which are almost equal to $\hat{X}_{S,f}$. But by virtue of both Corollary \ref{cc-dd'} and the properties of $\chi$ the value assigned to $\hat{X}_{S,f}$ is again well defined.\smallskip

{\bf Step 2.} Assume that the assignment \eqref{tw1_f2} extends by linearity to a (linear) functional $\omega$ on $\tilde{\Af}$. Then by Lemma \ref{lm-real} the functional is real, and hence by virtue of Lemma \ref{lem_positivity} it is a state on the algebra.\smallskip

{\bf Step 3.} The assignment $\hat{\Psi}\mapsto \mu_{\rm AL}(\Psi)$ is diffeomorphism invariant by virtue of the properties of the measure $\mu_{\rm AL}$ \cite{al-hoop}. 

Suppose that the action \eqref{dif_action} of a diffeomorphism $\varphi$ maps $\hat{X}_{C^d,f}$ to $\hat{X}_{C^{\prime d'},f'}$. The diffeomorphism covariance \eqref{I_cov} of the set of isolated boundary points means that the vertex $v$ of the cone $C^d$ is mapped to the vertex $v'$ of $C^{\prime d'}$. Therefore $\tilde{f}'(v')=\tilde{f}(v)$. On the other hand as stated by Lemma \ref{ineq} any two cones $C^{1}$ and $C^{d}$ for $d>1$ are diffeomorphicly inequivalent i.e. there is no diffeomorphism which maps $C^1$ onto $C^d$. Consequently, the value $\chi(d)$ associated to $C^d$ is equal to $\chi(d')$ associated to $C^{\prime d'}$. Finally by virtue of Lemma \ref{orient} if $C^d$ is externally oriented then $C^{\prime d'}$ is externally oriented also. All these mean that the assignment $\hat{X}_{C^d,f}\mapsto\omega(\hat{X}_{C^d,f})$ given by \eqref{tw1_f2} is diffeomorphism invariant. 

Note finally that $(i)$ the relation $\hat{X}_{S,f}\approx \hat{X}_{C,f'}$ is diffeomorphism invariant  and $(ii)$ if $S$ is not contained in any $d$-cone then for every analytic diffeomorphism $\varphi$ the image $\varphi(S)$ is not contained in any $d$-cone either.       

Hence if the assignment \eqref{tw1_f2} defines a linear functional then the functional is diffeomorphism invariant.\smallskip

{\bf Step 4.} Let us start this part of the proof with simplifying the notation: instead of $\hat{X}_{S_\alpha,f_\alpha}$, where $\alpha$ is a label or a multilabel, we will write $\hat{X}_\alpha$. Consequently, given $\hat{X}_\alpha$, we will refer to corresponding smearing functions as $f_\alpha$ and $\tilde{f}_\alpha$.      

Note now that it is not obvious at all that an extension of the assignment given by Theorem \ref{main} to a linear functional on $\tilde{\Af}$ is possible: the source of a potential obstacle is the fact that we assigned numbers to elements of the form $\hat{\Psi}\hat{X}_{1}\ldots\hat{X}_{n}$---these elements span the algebra $\tilde{\Af}$ but they {\em are not} linearly independent. Therefore we have to check whether the assignment respects linear dependence between the elements under consideration. It is easy to realize that there is no obstacle for the extension if for every sum of the elements such that
\begin{equation} \label{tw1_eq1}
\sum_{i=1} ^{n}\hat{\Psi}_{i}\hat{X}_{i1}\ldots\hat{X}_{ik_i}=0.
\end{equation}
the sum of the assigned numbers is zero:
\begin{equation}
\sum_{i=1} ^{n}\omega(\hat{\Psi}_{i}\hat{X}_{i1}\ldots\hat{X}_{ik_i})=0.
\end{equation}

Let us now divide the set $\cal T$ of all terms occurring in the sum \eqref{tw1_eq1} in a way convenient for further considerations: we distinguish three disjoint subsets of $\cal T$ as follows:
\begin{enumerate}
\item the subset $\cal F$ consists of terms of the form $\hat{\Psi}_i$ ($\cal F$ stands for functions);
\item the subset $\cal C$ consists of terms of the form $\hat{\Psi}_{i}\hat{X}_{i1}\ldots\hat{X}_{ik_i}$ such that {\em every} flux operator $\hat{X}_{ij}$ is based on a face contained in a $d$-cone;  
\item the subset $\cal O:={\cal T}\setminus({\cal F}\cup{\cal C})$ is a subset of all other terms.  
\end{enumerate}

Note now that
\[
{\rm span}\,\{\ a \ | \ a\in{\cal F}\ \}\ \cap\ {\rm span}\,\{\ a \ | \ a\in{\cal C} \ \text{or} \ a\in{\cal O}\ \} =0.
\]
It follows from this fact and  \eqref{tw1_eq1} that 
\begin{equation} \label{tw1_eq*}
\sum_{a\in {\cal F}}a=0.
\end{equation}
Therefore
\[
\sum_{a\in {\cal F}}\omega(a)=\sum_{a\in {\cal F}}\mu_{\rm AL}(a)=\mu_{\rm AL}(\sum_{a\in {\cal F}}a)=0.
\]
Because $\omega(a)=0$ for any $a\in{\cal O}$ our task is reduced to proving that
\[
\sum_{a\in {\cal C}}\omega(a)=0.
\]

By virtue of \eqref{tw1_eq1} and \eqref{tw1_eq*} 
\[
\sum_{a\in{\cal C}\cup{\cal O}}a=0,
\]
which means that for any cylindrical function $\Phi$ 
\begin{equation}
\sum_{a\in{\cal C}\cup{\cal O}}\pi_{\rm AL}(a)\Phi=0,
\label{pi_phi}
\end{equation}
where $\pi_{\rm AL}$ is the Ashtekar-Lewandowski representation\footnote{The representation is defined as follows
\[
\pi_{\rm AL}(\hat{\Psi})\Phi:=\Psi\Phi, \ \ \ \pi_{\rm AL}(\hat{X}_{S,f})\Phi:=-iX_{S,f}\Phi.
\]} of $\tilde{\Af}$ \cite{area}. Now we are going to choose some specific cylindrical functions $\Phi$ in the equation above to extract some relevant information about elements of $\cal C$ while neglecting at the same time all information about elements in $\cal O$. The construction of these functions proceeds as follows.   
   
Denote by $\cal S$ the set of all faces underlying the flux operators occurring in the sum \eqref{tw1_eq1}. It is not difficult to realize that there exists a set ${\cal S}_C=\{C_1,\ldots,C_m\}$ of $d$-cones of the following properties:
\begin{enumerate}
\item given a cone $C_i\in {\cal S}_C$ every $S\in{\cal S}$ is either contained in $C_i$ or it is not contained in closure of the maximal analytic extension of $C_i$;       
\item given $C_i$ there exist at least one $S\in{\cal S}$ such that $S$ is contained in $C_i$;     
\item given $C_i\neq C_j$, $C_i$ is not contained in the maximal analytic extension of $C_j$.
\end{enumerate}

Thanks to the property 1 above for every $C_i\in{\cal S}_C$ we can define a set $\check{C}_i$ as it was described at the very end of the previous section (see Equations \eqref{sss} and \eqref{dens-C}). Let 
\begin{equation} \label{tw1_cart}
\check{C}_{\times}:=\check{C}_1\times\ldots\times\check{C}_m 
\end{equation}
Given $(x_{1},\ldots,x_{m})\in\check{C}_{\times}$, we can find a graph ${\gamma}$ with edges $\{e_1,\ldots,e_m\}$ such that for every $i=1,\ldots,m$ 
\begin{enumerate}
\item $e_i\cap C_i=x_i$ and $x_i$ is the source of $e_i$; 
\item $e_i$ lies ``up'' the oriented cone $C_i$;  
\item for every $S\in{\cal S}$ not contained in $C_i$ either $e_i\cap S=\varnothing$ or $e_i$ is contained in $\overline{S}$ \footnote{Let us fix a face $S\not\subset C_i$. Then $x_i\not\in S$. Consider an analytic edge $e$ satisfying the conditions 1 and 2. Now according to Definition \ref{face} it can be adapted to $S$ i.e. divided into shorter edges such that each of them  satisfies (modulo orientation) one of the three conditions of Definition \ref{adapted}. The shorter edge containing $x_i$ can be chosen in such a way that it satisfies the condition $(i)$ or $(iii)$ of the definition. Call this edge $e_S$. Then $e_i:=\bigcap e_S$ where the intersection runs over all faces $S\in{\cal S}$ not contained in $C_i$.}.
\end{enumerate}
It is easy to see that for any cylindrical functions $\Phi$ compatible with $\gamma$ the sum over ${\cal C}\cup{\cal O}$ in \eqref{pi_phi} reduces to a sum over $\cal C$.

Consider now an operator $\hat{X}_{ik}$ appearing in an element of $\cal C$ and the face $S$ underlying the operator. Note that if $S\subset C_j\in{\cal S}_C$ then $I_{\partial S}\cap \check{C}_j=\varnothing$ and consequently 
\[
\pi_{\rm AL}(\hat{X}_{ik})\Phi=\pi_{\rm AL}(\hat{X}_{C_j,f'_{ik}})\Phi
\]      
provided $\hat{X}_{ik}\approx \hat{X}_{C_j,f'_{ik}}$ (see Lemma \ref{lem-X-ext}). This means that while evaluating \eqref{pi_phi} we can use appropriate operators based on the cones in ${\cal  S}_C$.

Now let us choose $\Phi$ of the form
\begin{equation}
\Phi(A)=\exp[\,i\sum_{j=1}^m  n_j h_{e_j}(A)\,], \ \ \ n_j\in\mathbb{R},
\label{phi-exp}
\end{equation}
where $h_{e_j}(A)\in U(1)$ is the holonomy of $A$ along the edge $e$. Before we will evaluate the sum at the l.h.s. of \eqref{pi_phi} on cylindrical functions just introduced let us make some preparatory remarks. Note first that it follows immediately from  Definition \ref{flux-op} of flux operators and the form of the Ashtekar-Lewandowski representation that if $\hat{X}_{ik}\approx \hat{X}_{C_j,f'_{ik}}$ then  
\[
\pi_{\rm AL}(\hat{X}_{ik})\Phi=\pi_{\rm AL}(\hat{X}_{C_j,f'_{ik}})\Phi=n_jf'_{ik}(x_{ik}) \Phi.
\] 
where $x_{ik}\equiv x_j$. 

Given an element $a_i=\hat{\Psi}_{i}\hat{X}_{i1}\ldots\hat{X}_{ik_i}\in {\cal C}$, it can happen that some flux operators constituting the element are based on faces contained in the same $d$-cone $C_j$. Let ${\cal C}_{\vec{l}}$, $\vec{l}=(l_1,\ldots,l_m)$, denotes a subset of $\cal C$ such that in each element of the subset there are $l_j$ flux operators based on a face contained in the cone $C_j$. Thus if $a_i$ belongs to ${\cal C}_{\vec{l}}$ then
\[
\pi_{\rm AL}(a_i)\Phi=\pi_{\rm AL}(\hat{\Psi}_{i}\hat{X}_{i1}\ldots\hat{X}_{ik_i})\Phi=\Psi_i n^{l_1}_1\ldots n^{l_m}_m f'_{i1}(x_{i1})\ldots f'_{ik_i}(x_{ik_i})\Phi,
\]      
where for all the flux operators constituting $a_i$ there holds an appropriate relation $\hat{X}_{ik}\approx \hat{X}_{C_j,f'_{ik}}$ which defines the identification $x_{ik}\equiv x_j\in \check{C}_j$.  

After these preparations it is easy to evaluate the sum at the l.h.s. of \eqref{pi_phi} on  \eqref{phi-exp}:
\begin{multline*}
\sum_{a\in{\cal C}\cup{\cal O}}\pi_{\rm AL}(a)\Phi=\sum_{\vec{l}}\sum_{a\in{\cal C}_{\vec{l}}}\pi_{\rm AL}(a)\Phi=\\=\sum_{\vec{l}}\sum_{i\ |\ a_i\in{\cal C}_{\vec{l}}}\Psi_i n^{l_1}_1\ldots n^{l_m}_m f'_{i1}(x_{i1})\ldots f'_{ik_i}(x_{ik_i})\Phi=\\=\sum_{\vec{l}}\,n^{l_1}_1\ldots n^{l_m}_m\Big(\sum_{i\ |\ a_i\in{\cal C}_{\vec{l}}}\Psi_i  f'_{i1}(x_{i1})\ldots f'_{ik_i}(x_{ik_i})\Big)\Phi=0.
\end{multline*}

The key observation now is that the last equation holds for all $n_j\in\R$ and it is a polynomial of real numbers $\{n_j\}$. Therefore  
\begin{equation}
\sum_{i\ |\ a_i\in{\cal C}_{\vec{l}}}\Psi_i {f'}_{i1}(x_{i1})\ldots {f'}_{ik_i}(x_{ik_i})=0
\label{psi-fff}
\end{equation}
where we got rid of $\Phi$ (note that $\Phi(A)\neq 0$ for every $A$). 

Recall that $x_{ik}\equiv x_j\in \check{C}_j$ if $\hat{X}_{ik}\approx \hat{X}_{C_j,f'_{ik}}$. These means that the product of the smearing functions in \eqref{psi-fff} is evaluated at a point belonging to $\check{C}_\times$. Since $(x_1,\ldots,x_m)$ is an arbitrary point of the set, Equation \eqref{psi-fff} holds on the whole $\check{C}_\times$. It follows from Corollary \ref{dens-C} that this set is a dense subset of ${C}_1\times\ldots\times{C}_m$. Therefore for every $j$ we can pass to the limit $x_{ik}\equiv x_j\to v_{j}$, where $v_{j}$ is the vertex of the cone $C_j$. For each smearing function $f'_{ik}$ the limit is equal to $\tilde{f}'_{ik}(v_{j})$ (see Definition \ref{def_fun}). Thus we obtain 
\[
\sum_{i\ |\ a_i\in{\cal C}_{\vec{l}}}\Psi_i \tilde{f}'_{i1}(v_{j})\ldots \tilde{f}'_{ik_i}(v_{j'})=0.
\]  
Integrating both sides of the equation above by means of the Ashtekar-Lewandowski measure and multiplying them by an appropriate number of factors $\chi(d)$ corresponding to the cones $\{C_j\}$ we obtain
\[
\sum_{a_i\in{\cal C}_{\vec{l}}}\omega(a_i)=0
\]    
which holds for every $\vec{l}$. This ends the proof.
\end{proof}

\section{Summary}

In this paper we constructed a class of diffeomorphism invariant states on the modified holonomy-flux algebra. Each state defines the GNS representation $\pi$ of the algebra on a Hilbert space $\mathcal{H}$, which as in the case of the standard state is isomorphic to $L^{2}(\overline{\mathcal{A}},d\mu_{AL})$. The representation is described by the following simple formulae
\begin{align*}
&\pi(\hat{\Psi})\Phi=\Psi\Phi,\\ 
&\pi (\hat{X}_{S,f})\Phi=-i{X_{S,f}(\Phi)}+\omega(\hat{X}_{S,f})\Phi
\end{align*}
for any $\Phi\in Cyl^\infty\subset L^{2}(\overline{\mathcal{A}},d\mu_{AL})$. These hold true for the GNS representation defined by the standard state (that is, the Ashtekar-Lewandowski representation) the only difference being that then $\omega(\hat{X}_{S,f})$ is zero for every flux operator.

As mentioned in the introduction the main idea of this paper was to verify whether the example of a non-standard state presented in \cite{LOST} can be generalized to higher dimensional base manifolds. As we saw the answer is in the affirmative. It has to be however emphasized that the result relies heavily on the following assumptions concerning  $(i)$ the structure group of the principal bundle underlying the construction $(ii)$ the construction of the modified holonomy-flux algebra $\tilde{\Af}$ and $(iii)$ the choice of the gauge transformations acting on the algebra:   
\begin{enumerate}
\item[(i)] the structure group is $U(1)$. Note that by virtue of the last property of the ansatz \eqref{new-st} for any flux operators $\omega([\hat{X}_{S,f},\hat{X}_{S',f'}])=0$. In the case of a (semisimple) noncommutative structure group we can expect that flux operators may be expressed as commutators of other flux operators which would mean that our ansatz would be reduced to the formulae defining the standard state. Moreover, the fact that in the $U(1)$ case the flux operators commute allowed us to use in the proof of the main theorem \ref{main} the function $\Phi$ \eqref{phi-exp} being an eigenvector of all the operators. This property of $\Phi$ made that part of the proof relatively easy. Note also that in the $U(1)$ case the smearing functions can be treated as real ones hence the limit \eqref{assign-2} could be used as a value of the non-standard states  on $\hat{X}_{S,f}$ invariant with respect to Yang-Mills gauge transformations. 
\item[(ii-a)] the support of each smearing function $f: S\to\R$ is noncompact in a rather specific way i.e. it is compact modulo isolated boundary points of the face $S$. Our construction does not seem to be valid in the case of the holonomy-flux algebra allowing smearing functions with any non-compact supports. 
\item[(ii-b)]  the faces (and graphs) are analytic. We used properties of analytic faces in the proof of the main theorem \ref{main} (more precisely: in the proof that the assignment \eqref{tw1_f2} can be extended by linearity to a functional on $\tilde{\Af}$) and also in the proof of Lemma \ref{orient}. Perhaps the construction of the states is still valid in the semianalytic case but then a proof of Theorem \ref{main} would be technically much harder.
\item[(iii)] the gauge transformations do not contain any non-analytic diffeomorphisms, which is consistent with the choice of analytic faces.
\end{enumerate}

Let us finally state that in our opinion the result presented in this paper does not weaken the significance of the uniqueness theorem of \cite{LOST}. It is true, of course, that the assumptions of the theorem are not general, but we learnt from this research that in order to obtain a non-standard state we have to $(i)$ replace the simple and natural assumptions of the theorem by ones which are not so simple and natural and $(ii)$ base the construction on some very special structures like vertices of cones where the differentiability of a face is broken. Such structures are pointlike and therefore it seems that they cannot play any role in physical models where the smallest structures are expected to be of the Planck length. 

\paragraph{Acknowledgments}  This work was partially supported by the Polish Ministerstwo Nauki i Szkolnictwa Wy\.zszego grants 1 P03B 075 29 and 182/NQGG/ 2008/0, by the Foundation for Polish Science grant "Master" {and a Travel Grant from the QG research networking programme of the European Science Foundation.}

\appendix

\section{Graph not adaptable to a submanifold \label{counter}}

Let $\Sigma=\R^3$ and let 
\[
\mathbf{S}=\{\ (x,y,z)\in\R^3 \ | \ 0<x<2, \ -1< y <1, z=0 \ \}.
\]  
Denote by $B(p,r)$  a {\em closed} ball in $\R^3$ of center $p$ and radius $r>0$ and for $n\in\mathbb{N}$ define
\[
B_n:=B\Big(\Big(\frac{1}{n},0,0\Big),\frac{1}{4}\Big(\frac{1}{n}-\frac{1}{n+1}\Big)\Big)
\]
Clearly, $B_n\cap B_{n'}=\varnothing$ for $n\neq n'$. The set
\[
S:=\mathbf{S}\setminus \bigcup_{n=1}^\infty B_n
\]
is an embedded analytic submanifold of $\R^3$. Then the graph $\gamma=\{e\}$, where 
\[
e:=\{(x,y,z)\in \R^3 \ | \ 0\leq x\leq 1, \ y=z=0\ \}
\]
is an analytic edge, is  not adaptable to $S$ in the sense of Definition \ref{adapted}. Therefore according to Definition \ref{face} $S$ is not a face.  

Consider now the set
\[
S':=\mathbf{S}\setminus \{\ (1/n,0,0)\in\R^3\ | \ n\in\mathbb{N} \ \}.
\] 
$S'$ is again an embedded analytic submanifold of $\R^3$ and the graph $\gamma$  is adaptable to $S'$---it is in fact adapted to $S'$ since $e\subset\overline{S'}$. It is easy to realize that $S'$ is a face in the sense of  Definition \ref{face}. Moreover, the set $I_{\partial S'}$ of its isolated boundary points is infinite. Thus $S'$ is an example of a face with infinite isolated boundary points.     
   
\section{Inequivalence of $C^1$ and $C^d$ with $d>1$ with respect to diffeomorphisms }

\begin{lm}
There is no diffeomorphism on $\R^D$ which maps a cone $C^1$ into a cone $C^d$ with $d>1$ or a cone $C^d$ with $d>1$ into a cone $C^1$.    
\label{ineq} 
\end{lm}

\begin{proof}
It is enough to prove the lemma in the case of cones given by \eqref{cone}. The proof we are going to present is based on an analysis of curves ingoing to the cones through their vertex $v=(0,\ldots,0)$. In the proof we will use notation introduced in Subsection \ref{cone-like} below Lemma \ref{orient}. 

To prove the lemma in the case of cones \eqref{cone} it is enough to justify the following three facts:         
\begin{enumerate}
\item vectors in $T_v\R^D$ of non-zero $D$-th components generated by curves ingoing to $C^1$ through $v$ span $T_v\R^D$. 
\item vectors in $T_v\R^D$ of non-zero $D$-th components generated by curves ingoing to $C^d$, $d>1$, through $v$ span a one dimensional subspace of $T_v\R^D$. 
\item if a diffeomorphism $\varphi$ maps $C^1$ into $C^d$ ($d>1$) or $C^d$ into $C^1$ then 
\begin{enumerate}
\item it preserves the vertex $v$ and maps a curve $\kappa$ ingoing to $C^1$ (or $C^d$) through $v$ into a curve $\varphi(\kappa)$ ingoing to $C^d$ (or $C^1$) through its vertex,
\item if the $D$-th component of $\dot{\kappa}$ is non-zero then the $D$-th component of $\varphi'\dot{\kappa}\in T_v\R^D$ is non-zero also.            
\end{enumerate}
\end{enumerate}
Indeed---all these together mean that if $\varphi$ maps $C^1$ into $C^d$ or {\em vice versa} then the tangent map $\varphi'$ is non-linear which cannot be true. Now let us show that the three statements above are true.

To prove the first statement consider a $D\times D$-matrix
\[
(\Lambda_i{}^j)=
\begin{pmatrix}
1 & 0 & \ldots & 0 & 0 & 1\\
0 & 1 & \ldots & 0 & 0 & 1\\
\vdots& \vdots & \ddots&\vdots & \vdots & \vdots\\
0& 0 & \ldots & 1 & 0 & 1\\
0& 0 & \ldots & 0 & 1 & 1\\
0 & 0 & \ldots& 0 & -1  & 1
\end{pmatrix}
\]
Obviously, every curve
\[
\tau\mapsto \kappa_i(\tau)=(\Lambda_i{}^1\tau,\ldots,\Lambda_i{}^D\tau)\in\R^D
\]
is ingoing to $C^1$ through $v$, the $D$-th component of each $\dot{\kappa}_i$ is non-zero and the vectors span $T_v\R^D$ (note that the determinant of $(\Lambda_i{}^j)$ is non-zero).     

To show that the second statement is correct consider a curve $\kappa:\tau\to(x^i(\tau))$ ingoing to $C^d$, $d>1$,  through its vertex $v$ such that $\dot{\kappa}^D>0$. Then we can use the coordinate $x^D$ to reparametrize the curve in the neighborhood of $v$---there exists $a>0$ and a function $z\mapsto \tau(z)$ defined on $]-a,a[$ such that $x^D(\tau(z))=z$ on $]-a,a[$. Then    
\[
]-a,a[\ni z\mapsto{\kappa}'(z):=\kappa(\tau(z))
\]     
is a well defined $C^1$-class curve in $\R^D$. Obviously,
\[
\dot{\kappa}=\frac{dx^D}{d\tau}\Big|_{\tau=0}\dot{\kappa}'=\dot{\kappa}^D\dot{\kappa}'.
\] 
A coordinate expression of the curve $\kappa'$ for $z\in]0,a[$ reads
\[
\kappa'(z)=(x^1(z),\ldots,x^{D-1}(z),x^D(z))=(z^dy^1(z),\ldots,z^dy^{D-1}(z),z)
\]               
where
\[
y^i(z):=\frac{x^i(z)}{z^d}, \ \ \ i=1,\ldots,D-1.
\] 
These functions satisfy
\[
\sum_{i=1}^{D-1}(y^i(z))^2=1
\]
(see the formula \eqref{cone}) and are bounded therefore. Since $\kappa'(0)=(0,\ldots,0)$ we have for $i=1,\ldots,D-1$
\[
(\dot{\kappa}')^i=\frac{dx^i}{dz}\Big|_{z=0}=\lim_{z\to 0^+}\frac{z^dy^i(z)-0}{z}=\lim_{z\to 0^+}z^{d-1}y^i(z)=0.
\] 
Hence
\[
\dot{\kappa}'=(0,\ldots,0,1) 
\]
and consequently
\[
\dot{\kappa}=(0,\ldots,0,\dot{\kappa}^D).
\]
This proves the second statement.

To prove the third statement assume for definiteness that a diffeomorphism $\varphi$ maps $C^1$ into $C^d$. Note that $\varphi$ has to preserve the point $v=(0,\ldots,0)$ being the vertex of both cones. Consequently, an open neighborhood\footnote{An open neighborhood $U$ of the vertex $v$ in $C^d$, $d\geq 1$, is an open subset $U$ of $C^d$ such that $U=U_0\setminus \{v\}$, where $U_0$ is a subset of the closure $\overline{C^d}$ such that $v\in U_0$.} of $v$ in $C^1$ has to be mapped onto an open neighborhood of $v$ in $C^d$ and $\varphi$ maps a curve $\kappa$ ingoing to $C^1$ through $v$ into $\varphi(\kappa)$ ingoing to $C^d$ through $v$. In particular this means that 
\[
(\varphi'\dot{\kappa})^D=(\dot{\varphi(\kappa)})^D\geq 0
\]
(see \eqref{Dth}).

Consider then a curve $\kappa$ ingoing to $C^1$ through its vertex $v$ such that $\dot{\kappa}^D>0$. Suppose that $(\varphi'\dot{\kappa})^D=0$. Then we can use a $D$-dimensional rotation $R$ around the $x^D$-axis to map $\varphi'\dot{\kappa}$ into $-\varphi'\dot{\kappa}$, and then $\varphi^{\prime-1}$ to map the latter vector into $-\dot{\kappa}$. Thus $(\phi'\dot{\kappa})^D<0$, where $\phi:=\varphi^{-1}\circ R\circ \varphi$. 

But $R$ preserves the cone $C^d$ and therefore the diffeomorphism $\phi$ maps an open neighborhood of $v$ in $C^1$ onto an open neighborhood of $v$ in $C^1$. Thus $\phi(\kappa)$ is again ingoing to $C^1$ through $v$ and $(\phi'\dot{\kappa})^D\geq 0$ (see again \eqref{Dth}) which contradicts the previous conclusion. This means that the supposition $(\varphi'\dot{\kappa})^D=0$ cannot be true and  we are left with
\[
(\varphi'\dot{\kappa})^D>0.
\]
Thus the third statement is shown to be correct.
\end{proof}

\end{document}